\documentclass[aps,pra,superscriptaddress,12pt,nofootinbib,longbibliography]{revtex4-2}

\makeatletter
\renewcommand{\p@subsection}{}
\renewcommand{\p@subsubsection}{}
\makeatother

\usepackage{graphicx} 
\usepackage{dcolumn}
\usepackage{bm}        
\usepackage{amsmath,amssymb,amsfonts,amsthm}
\usepackage{mathtools}
\usepackage[normalem]{ulem} 

\newtheorem{theorem}{Theorem}[section]
\newtheorem*{theorem*}{Theorem}

\newtheorem{proposition*}{Proposition}

\newtheorem*{lemma*}{Lemma}
\newtheorem{corollary}[theorem]{Corollary}

\usepackage[mathscr]{eucal}
\usepackage{color}
\usepackage{xcolor}
\usepackage{tikz}
\usetikzlibrary{matrix,arrows,decorations.pathmorphing}
\usepackage[utf8]{inputenc}

\usepackage{xparse}
\usepackage{hyperref}
\usepackage{cleveref}

\usepackage{listings}
\usepackage{color} 
\definecolor{mygreen}{RGB}{28,172,0} 
\definecolor{mylilas}{RGB}{170,55,241}
\usepackage{xcolor}

\usepackage[normalem]{ulem} 
\usepackage{physics}
\providecommand{\ignore}[1]{}

\newif\ifcmnt
\cmnttrue
\ifdefined\cmntsoff\cmntfalse\fi

\ifcmnt
    \providecommand{\aucmnt}[1]{#1}

\else
    \providecommand{\aucmnt}[1]{}

\fi

\newcommand{\SI}[1]{\;\textrm{#1}}

\newcommand{\cD}{\mathcal{D}}

\newcommand{\LO}{\textrm{LO}}

\newcommand{\lang}{\langle}
\newcommand{\rang}{\rangle}

\begin{document} 
\bibliographystyle{unsrt}
\title{Broadband pulsed quadrature measurements with calorimeters}
\author{Ezad Shojaee}
\altaffiliation[Current address: ]{IonQ, College Park, MD 20740}
\affiliation{National Institute of Standards and Technology, Boulder, Colorado 80305, USA}
\affiliation{Department of Physics, University of Colorado, Boulder, Colorado, 80309, USA}
\author{James R. van Meter}
\affiliation{National Institute of Standards and Technology, Boulder, Colorado 80305, USA}
\author{Karl Mayer}
\affiliation{Quantinuum, Broomfield, CO 80021}
\author{Scott Glancy}
\affiliation{National Institute of Standards and Technology, Boulder, Colorado 80305, USA}
\author{Emanuel Knill}
\affiliation{National Institute of Standards and Technology, Boulder, Colorado 80305, USA}
\affiliation{Center for Theory of Quantum Matter, University of Colorado, Boulder, Colorado 80309, USA}
\date{\today}

\begin{abstract}
  A general one-dimensional quantum optical mode is described by a
  shape in the time or frequency domain. A fundamental problem is to
  measure a quadrature operator of such a mode. If the shape is narrow
  in frequency this can be done by pulsed homodyne detection, in which
  the mode and a matched local oscillator (LO) interfere on a
  beamsplitter, whose output ports are monitored by
  photo-detectors. The quadrature value is proportional to the
  difference between the photo-detectors' signals. When the shape of
  the mode is broad in frequency, the lack of uniform response of the
  detectors across the spectrum prevents direct application of this
  technique. We show that pulsed homodyne detection can be generalized
  to \emph{broadband pulsed (BBP) homodyne detection} setups with
  detectors such as calorimeters that detect total energy instead of
  total number of photons. This generalization has applications in
  quantum measurements of femtosecond pulses, and, speculatively,
  measurements of Rindler modes to verify the temperature of Unruh
  radiation. Like pulsed homodyne detection, BBP homodyne detection
  requires choosing the LO pulse such that the subtracted signal
  approaches the desired quadrature measurement for large LO
  amplitudes. A distinctive feature of the technique is that the LO
  pulse does not belong to the mode of the quadrature being measured.
  We analyze how the implemented measurement approaches an ideal
  quadrature measurement with growing LO amplitude. We prove that the
  moments of the measurement converge to the moments of the quadrature
  and that the measurement distributions converge weakly.
\end{abstract}

\maketitle

\section{Introduction}
Homodyne detection is an indispensable part of many quantum optics
experiments. In homodyne detection one measures quadrature operators
of an optical signal of interest. It is widely used in both the
frequency and the time domains; for classic examples
see~\cite{Breitenbach97, Faridani1993}. In homodyne detection one
interferes the input signal with a high-amplitude, mode-matched local
oscillator (LO) on a balanced beamsplitter and measures the light
exiting the two output ports of the beamsplitter with high-efficiency
detectors such as photodiodes (Fig.~\ref{fig:fig1}).  The quadrature
measurement outcome is obtained by scaling the difference between the
detector outputs. The measurement outcome approaches that of an ideal
quadrature measurement in the limit of large LO amplitude.  Homodyne
detection can be performed with an always-on LO (continuous wave
homodyne) or with a pulsed LO (pulsed homodyne).  In continuous wave
homodyne, detectors continuously record intensity. To determine a
specific quadrature measurement outcome, the intensity difference is
integrated against the shape in time defining the quadrature, as long
as this shape's spectrum is within the bandwidth of the detector. Many
quadratures can be simultaneously measured in this scheme.  In pulsed
homodyne, the LO is pulsed with a shape matching that of the
quadrature of interest. Only the quadrature matching the LO is
measured, but the detectors need not be time-resolving and the LO
power can be concentrated where it is needed.  It has the benefits of
simplicity and is suitable for the case where the quadrature of
interest is known ahead of time. Pulsed homodyne is required when
using slow detectors such as calorimeters.

Standard pulsed homodyne detection requires the detector outputs to be
proportional to the total number of photons arriving during the
pulse's time interval. This restricts the bandwidth of the LO pulse
and of the optical signal during the pulse's time interval to a
spectral band over which the detectors have nearly uniform
efficiency. To approach a single-shot quadrature measurement, the
efficiency needs to be close to unity over the spectral band.  There
are no detectors with a wide enough spectral band to measure
quadratures of optical signals with octave-spanning bandwidth such as
those associated with femtosecond pulses. Such octave-spanning
quadratures are also of interest in studies of the Unruh
effect~\cite{RevModPhys}, which arises for accelerating observers due
to the non-zero temperature of Rindler modes.  Rindler modes are
naturally extremely broadband in an inertial frame.

One type of detector with the potential for high-efficiency detection
over a broad frequency band is a calorimeter. Rather than counting
photons, calorimeters record total energy. Calorimeters are widely
used and can achieve near-unit efficiency for measuring energy across
the spectrum. For example transition edge sensors (TES) directly
detect the change of temperature due to incident light on a
superconducting island at the transition
temperature~\cite{irwin1995application,irwin1995quasiparticle} and can
be designed to have high efficiency over a wide bandwidth. In 1998, Ref.~\cite{cabrera1998detection} demonstrated broadband detection
over a \(3.5\) octave bandwidth from
\(0.3\SI{eV}\) to \(3.5\SI{eV}\) with an energy resolution of about
\(0.15\SI{eV}\). Some of these parameters have since been improved. For example,
Ref.~\cite{hattori2022optical} achieved \(60\SI{\%}\) efficiency with an energy resolution of \(0.067\SI{eV}\) in a band around \(0.8\SI{eV}\).

To enable the use of detectors such as calorimeters for measuring
specific quadratures of broadband modes, we introduce \emph{broadband
  pulsed (BBP) homodyne detection}. BBP homodyne is based on detectors
that record observables, such as total energy, that are expressed as
linear combinations of photon numbers in the modes of an orthogonal
mode basis. As in standard pulsed homodyne, the BBP homodyne
measurement outcome is the scaled difference between the detector
outputs.  We show that arbitrary quadratures expressible in this mode
basis can be measured with a high-amplitude LO whose pulse shape is
determined by the quadrature of interest. A feature of BBP homodyne is
that the mode of the LO is typically different from the mode whose
quadrature one wishes to measure.  BBP homodyne is distinguished from methods
that allow for the parallel measurement of narrow-band quadratures
across a broad spectral band, such as those described in
Ref.~\cite{shaked2018lifting}. Such parallel quadrature measurements
can in principle be combined in software to obtain measurement
outcomes for specific broadband modes of interest, at the cost of the
additional resources and high LO energy required for the simultaneous
measurements of many narrow-band quadratures.  

Theoretical treatments of pulsed homodyne detection argue that the
homodyne measurement outcomes correspond to an observable that differs
from the quadrature to be measured by an operator proportional to the
inverse of the LO amplitude.  For well-behaved states, such as those
with bounded quadrature moments, this implies that the moments of the
homodyne measurement outcome converge to those of the quadrature.  We
show that for BBP homodyne, the moments of the homodyne measurement
outcomes also converge to those of the quadrature.  For homodyne
detection in general, the difference between a moment of the homodyne
measurement outcomes and that of the quadrature depends on the
properties of the state and the order of the moment, assuming the LO
amplitude is finite.  This difference can be arbitrarily large at any
finite LO amplitude. Therefore, the convergence of the measurement
outcome distribution is not addressed by considerations of the moments
alone. An early study of the finite-LO amplitude behavior and the
convergence of the outcome distribution is
Ref.~\cite{braunstein1990homodyne}, where these outcome distributions
are derived and compared to quadrature outcome distributions for
superpositions of two coherent states. Several more detailed studies
investigating the finite-LO amplitude measurement operators in the
context of operational homodyne detection
followed~\cite{Banaszek,Sanders}.  Beyond convergence of moments, it
desirable to have weak convergence of measurement distributions, which
can be defined as convergence of expectation values of continuous
bounded functions of the measurement outcomes.  A general theory
relating the convergence of moments to weak convergence was developed
and applied by J. Kiukas and P. Lahti~\cite{kiukas2008moment_jmo,
  kiukas2008moment}.  Here, we apply this theory to show that weak
convergence holds also for BBP homodyne.

To present our results, we begin with a review of standard homodyne
detection in Sect.~\ref{sec:bbh}. We introduce BBP homodyne detection
in Sect.~\ref{calorimeter}. We then investigate the convergence of BBP
homodyne measurements to ideal quadrature measurements. Convergence of
moments and weak convergence of the measurement distribution is
established in Sect.~\ref{sec:wconv}. We conclude with a discussion
of BBP and our results in Sect.~\ref{sec:discussion}.

\section{Homodyne detection}
\label{sec:bbh}
Standard pulsed homodyne detection is a way to measure a quadrature of
a mode of a field. Typical applications involve electromagnetic fields
whose excitations are photons, and we use terminology accordingly.
Such a field may be described by annihilation and creation operator
fields in momentum space, see~\cite{leonhardt:qc1997a} for a
pedagogical treatment. We denote generic modes by $a,b,\ldots$ and the
corresponding mode (annihilation) operators by
$\hat{a},\hat{b},\ldots$ with or without indices. Mode operators
satisfy $\left[\hat{a},\hat{a}^{\dagger}\right]=1$. In the physical
field's state space, $\hat a$ annihilates the vacuum, and
$\hat a^{\dagger}$ applied to the vacuum creates a photon in mode $a$
with a particular wavefunction in momentum space.  The number operator
for generic mode $a$ is $\hat{n}_{a}=\hat{a}^{\dagger}\hat{a}$ and
counts the number of photons in this mode. We denote the vacuum state
of a mode by $\ket{0}$. The state space of a mode is the Hilbert space
spanned by the number states
$\ket{n}=\left(\hat{a}^{\dagger}\right)^{n}\ket{0}/\sqrt{n!}$.

We use the convention that generalized quadratures of mode $a$ are
defined by operators
$\hat q_{a,\alpha} = -i (\alpha \hat{a}^{\dagger}- \alpha^{*}\hat{a})$
where $\alpha$ is a complex number. Conventional, canonically
conjugate $\hat{x}$ (``position'') and $\hat{p}$ (``momentum'')
operators associated with the mode when viewed as a quantum harmonic
oscillator are $\hat{x}= \hat{q}_{a,i/\sqrt{2}}$ and
$\hat p = \hat q_{a,-1/\sqrt{2}}$. A coherent state is a normalized
state $\ket{\beta}$ that satisfies
$\hat{a}\ket{\beta} = \beta\ket{\beta}$.  This identity determines
\(\ket{\beta}\) up to a global phase.  Having fixed the normalized
vacuum state $\ket{0}$, we fix the phase by requiring that
$\bra{\beta}\ket{0}$ is positive real. This is equivalent to requiring
that \(\ket{\beta}\) has real, positive coefficient on \(\ket{0}\)
when expressed in the number basis. This coefficient is necessarily
\(e^{-|\beta|^{2}/2}\).  The overlap formula for coherent states is
then given by
\begin{align}
\braket{\beta}{\alpha} &= e^{-(\vert \beta \vert^{2} + \vert
\alpha \vert^{2} - 2 \alpha \beta^{*})/2}.
\end{align}
  
Displacement operators for mode $a$ are given by
$\hat{D}_{\beta}=e^{i\hat{q}_{a,\beta}}$, and we say that \(\hat{q}_{a,\beta}\) generates
  the displacement \(\hat{D}_{\beta}\).
They satisfy
$\hat{D}_{-\beta}=\hat{D}_{\beta}^{\dagger}$ and transform the mode operator according to 
\begin{align}
\hat{D}_{-\beta} \hat{a} \hat{D}_{\beta} &= \hat{a} + \beta.
\end{align}
They displace coherent states as follows:
\begin{align}
\hat{D}_{\beta}\ket{\alpha}&=e^{i\Im{\alpha^{*}\beta}}\ket{\alpha+\beta}.
\end{align}
The phase on the right-hand side can be calculated from
$\hat{D}_{-\alpha-\beta} \hat{D}_{\beta}\hat{D}_{\alpha}=e^{i\Im{\alpha^{*}\beta}}$ and is given by
twice the signed area of the triangular path followed in going from $0$ to $\alpha$ to $\alpha+\beta$ and back to $0$. For example, see Eq. (2) of Ref.~\cite{leibfried2003experimental}. 

Two modes $a$ and $b$ are orthogonal if
$[\hat{a},\hat{b}^{\dagger}]=0$. This implies that the two associated
photon wavefunctions are orthogonal and that polynomials of $\hat{a}$
and $\hat{a}^{\dagger}$ commute with polynomials of $\hat{b}$ and
$\hat{b}^{\dagger}$. The joint state space of the two modes can be
represented as the tensor product of the state spaces of each mode.
For our analysis, we consider finitely many modes at a time. A family
of orthogonal modes $\{a_{k}\}_{k=1}^{N}$ is associated with a vector
of mode operators $\hat{\bm{a}}=(\hat{a}_{1},\ldots,\hat{a}_{N})$
satisfying $[\hat{a}_{k},\hat{a}_{l}^{\dagger}]=\delta_{kl}$. The
joint vacuum of the modes is denoted by $\ket{\bm{0}}$, and its
density matrix is $\bm{0}$. The joint state space of the modes is
spanned by Fock states $\ket{\bm{n}}$, where
$\bm{n}=(n_{1},\ldots,n_{N})$ is the vector of mode occupation
numbers. That is, $n_{k}$ is the number of photons in mode $k$.

General mode operators in the system defined by the modes $\bm{a}$ can be associated with complex, normalized amplitude vectors $\bm{\beta}=(\beta_{1},\ldots,\beta_{N})$, $\sum_{k=1}^{N}|\beta_{k}|^{2}=1$. The operator $\hat{a}_{\bm{\beta}}=\sum_{k=1}^{N}\beta_{k}\hat{a}_{k}$ satisfies the defining properties of mode annihilation operators, namely
$\hat{a}_{\bm{\beta}}\ket{\bm{0}}=0$ and $\left[\hat{a}_{\bm{\beta}},\hat{a}_{\bm{\beta}}^{\dagger}\right] =
1$. As a result, generalized quadratures, number operators and states of mode $a_{\bm{\beta}}$ are well-defined.

Generalized quadratures of the modes $\bm{a}$ are defined
as
\begin{align}
\hat{q}_{\bm{a},\bm{\alpha}} &= \sum_{k=1}^{N}-i (\alpha_{k} \hat{a}_{k}^{\dagger} - \alpha_{k}^{*}\hat{a}_{k}).
\end{align}
The quadrature is normalized if \(\sum_{k}|\alpha_{k}|^{2}=1/2\).
The corresponding displacement operators are
\begin{align}
\hat{D}_{\bm{a}, \bm{\alpha}}&=e^{i\hat{q}_{\bm{a},\bm{\alpha}}}\nonumber\\
        &= \prod_{k}e^{i\hat{q}_{a_{k},\alpha_{k}}}.
\end{align}
The coherent states of modes $\bm{a}$ are 
\begin{align}
\ket{\bm{\alpha}} &= \otimes_{k=1}^{N}\ket{\alpha_{k}}_{a_{k}}\nonumber\\
&= \hat{D}_{\bm{a},\bm{\alpha}} \ket{\bm{0}}, 
\end{align}
where we use labels on the kets to specify the subsystem in a tensor
product that the state belongs to. The algebraic properties of coherent states and displacement operators generalize accordingly.

\begin{figure}
  \includegraphics[scale=0.47]{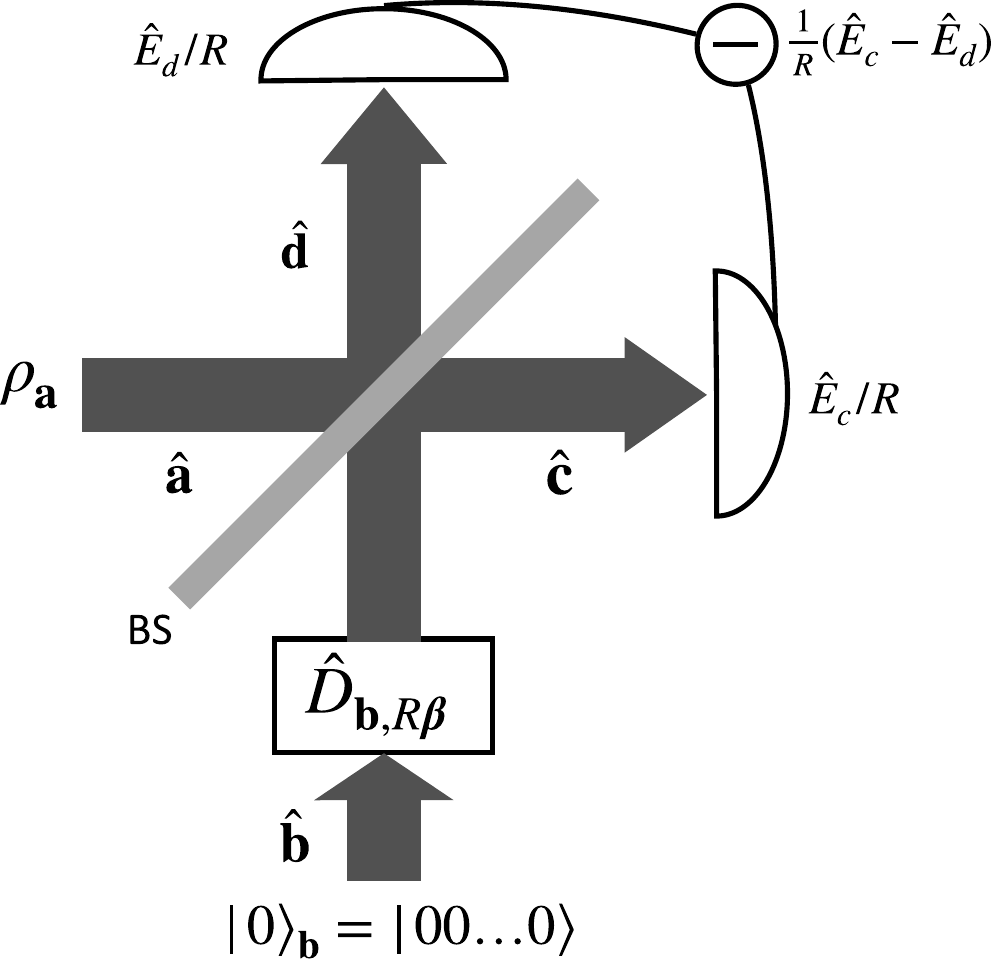}
  \centering
  \caption{Multi-mode homodyne measurement configuration.  The
    signal modes enter from the left on modes  \({\bm{a}}\) with mode
    operators \(\hat{\bm{a}}\).  The LO
    modes enter from the bottom on modes \({\bm{b}}\). To simplify
    calculations, the LO modes are initially in vacuum, and the LO
    coherent state is explicitly prepared by the displacement operator
    \(\hat{D}_{{\bm{b}},R\beta}\). Modes \(\bm{{a}}\) and \({\bm{b}}\) are
    combined on a balanced beam splitter (BS). The beam splitter's
    outgoing modes are \({\bm{c}}\) and \({\bm{d}}\). They are measured
    with photon counters (standard pulsed homodyne) or calorimeters
    (BBP pulsed homodyne). In both cases the observable associated
    with the detector is of the form
    \(\hat{E}=\sum_{k}\omega_{k}\hat{f}_{k}^{\dagger} \hat{f}_{k}\)
    with \(\bm{f}=\bm{c}\) or \(\bm{d}\). For standard pulsed
    homodyne, \(\omega_{k}=1\) for all \(k\). For BBP homodyne,
    \(\omega_{k}\) is the energy of photons in the \(k\)'th mode. The
    homodyne measurement result is determined by subtracting one
    detector's output from the other and rescaling the result by a
    factor of \(1/R\). In general, \(\omega_{k}\) can be arbitrary
    positive weights.}
  \label{fig:fig1}
\end{figure}

The configuration for pulsed homodyne is shown in Fig.~\ref{fig:fig1},
where for standard pulsed homodyne, the dectors count total number of
photons. The full state space of the signal is that of $N$ modes
$\bm{a}$ with mode operators $\hat{\bm{a}}$. The goal is to measure
the generalized quadrature $\hat{q}_{\bm{a},\bm{\beta}}$. We refer to
the quadrature to be measured as the target quadrature. For this
purpose an LO is introduced on a family of modes $\bm{b}$, each of
which is matched to the corresponding signal mode. To approximately
measure a generalized quadrature of the modes, the LO state is the
coherent state $\ket{R\bm{\beta}}_{\bm{b}}$, where $R$ is a large real
number. It is convenient to explicitly introduce the displacement
operator $\hat{D}_{\bm{b},R\bm{\beta}}$ that makes this state from
vacuum, as shown in the figure. The initial state consists of the
signal state $\rho_{\bm{a}}$ on the signal modes $\bm{a}$, and vacuum
$\bm{0}_{\bm{b}}$ on the LO modes $\bm{b}$.  The signal modes and the
matching LO modes are combined on a balanced beamsplitter. The
outgoing modes are labeled $\bm{c}$ and $\bm{d}$. With Heisenberg
evolution, we can express the outgoing mode operators in terms of the
original signal and the pre-displacement LO modes as follows:
\begin{align}
\label{eq:ckdk}
\hat{\bm{c}} &= \frac{1}{\sqrt{2}}\left(\hat{\bm{a}} + \hat{\bm{b}} + R\bm{\beta}\right),\nonumber\\
\hat{\bm{d}} &= \frac{1}{\sqrt{2}}\left(\hat{\bm{a}}-\hat{\bm{b}} - R\bm{\beta}\right).
\end{align}
Here we made a particular sign and phase choice for the balanced beamsplitter. The homodyne configuration is completed by photon counting on each of the two outgoing arms of the beamsplitter, subtracting the counts obtained, and dividing by $R$. To describe this measurement, we subtract the scaled total photon number operators for the arms to obtain the measurement operator
\begin{align}
\hat h &= \frac{1}{R}\left(\hat n_{\bm{c}}- \hat n_{\bm{d}}\right) \nonumber\\
&=\frac{1}{R}\left(\hat{\bm{c}}^{\dagger}\cdot\hat{\bm{c}} -
\hat{\bm{d}}^{\dagger}\cdot\hat{\bm{d}}\right)\nonumber\\
&= \bm{\beta}\cdot \hat{\bm{a}}^{\dagger} + \bm{\beta}^{*}\cdot\hat{\bm{a}} + \frac{1}{R}\left(
\hat{\bm{a}}^{\dagger}\cdot \hat{\bm{b}} +
                                                            \hat{\bm{b}}^{\dagger}\cdot \hat{\bm{a}}\right).
\end{align}
Here, the inner product between two vectors is denoted by ``$\cdot$'' and complex conjugation and dagger is applied element-wise. Intuitively, for large $R$ the summand with a factor of $\frac{1}{R}$ can be neglected, so $\hat h$ approaches a measurement of the target quadrature $\hat{q}_{\bm{a},i\bm{\beta}}$. 

Standard pulsed homodyne requires two ideal photon counters. In
practice, noisy, high-efficiency photon counters are used. If the
noise increases at a sublinear rate with photon number, a good
quadrature measurement can still be obtained by increasing the LO
amplitude $R$. High-efficiency photo-diodes have this property even
though they are unable to resolve photon numbers. Another effect that
needs to be considered is that the deviation of $\hat h$ from the
target quadrature measurement is affected by the state of photons in
modes orthogonal to that of the target quadrature. The amplitude $R$
needs to be increased to make the signal from such photons
negligible. Finally, the efficiency of the photon counters needs to be
uniformly high for all the relevant modes. Most photon counters have
high efficiency in a relatively narrow frequency band, which limits
the range of energies associated with the modes $a_{k}$. Calorimeters
have the potential to overcome this limitation but measure the total
energy $\sum_{k=1}^{N}\omega_{k}\hat n_{a_{k}}$ rather than the total
photon number.

\section{Quadrature measurement with calorimeters} \label{calorimeter}
An ideal calorimeter measures the total energy of the light. If we consider $N$ modes $\bm{a}$, then a calorimeter measures the energy operator 
\begin{equation}
\hat{E} = \sum_{k=1}^{N} \omega_{k} \hat{a}_{k}^{\dagger} \hat{a}_{k}.
\end{equation} 
where $\omega_{k}$ is the energy of a photon in mode $k$. One can
think of a calorimeter as a device performing measurement of a
weighted sum of photon numbers in a family of modes, where the weights
are the energies of the modes. Our treatment does not require
$\omega_{k}$ to be the energy of a mode $k$. It can be an arbitrary
real, positive weight. If $\omega_k=1$ for all $k$, $\hat{E}$ measures
the total photon number.  We therefore refer to the vector
$\bm{\omega}$ as mode weights.

Consider the homodyne setup with calorimeters instead of photon counters at the outgoing arms of the balanced beamsplitter.  The 
the measurement operator corresponding to subtracting the calorimeter measurement results and dividing by $R$ is 
\begin{align}
\frac{1}{R}\Delta \hat{E} & = \frac{1}{R}\sum_{k=1}^{N} \omega_{k} \big( \hat{c}_{k}^{\dagger} \hat{c}_{k} - \hat{d}^{\dagger}_{k} \hat{d}_{k} \big) \nonumber \\
   & = \sum_{k=1}^{N} \omega_{k} ( \hat{a}_{k}^{\dagger} \beta_{k} + \beta_{k}^{*} \hat{a}_{k} ) + \frac{1}{R}\sum_{k=1}^N \omega_{k} ( \hat{a}_{k}^{\dagger} \hat{b}_{k} + \hat{b}_{k}^{\dagger} \hat{a}_{k} ). \label{eq:calop}
\end{align}
Here all operators represent incoming mode operators according to
Heisenberg evolution. In particular \(\Delta\hat{E}\) when so
expressed depends on the LO displacement and therefore on \(R\).
Assuming that the contribution from terms multiplied by $\frac{1}{R}$
is negligible, this approaches a measurement of the target generalized
quadrature $\hat{q}_{\bm{a}, i(\bm{\omega} * \bm{\beta})}$, where
``*'' denotes the element-wise product defined by
$(\bm{\gamma}*\bm{\gamma'})_{k} = \gamma_{k}\gamma'_{k}$.  From now
on, we omit the mode label on operators such as $\hat{q}$ when the
modes are clear from context. For example, we write
\(\hat{q}_{\bm{\alpha}}\) for \(\hat{q}_{\bm{a},\bm{\alpha}}\). If we
wish to measure the generalized coordinate $\hat{q}_{\bm{\alpha}}$, we
choose $\bm{\beta} = -i \bm{\alpha}*(1/\bm{\omega})$, where
$1/\bm{\omega}$ is the vector with $k$'th entry $1/\omega_{k}$.  The
required LO amplitude is
\(R\bm{\beta}=-i \bm{\alpha}*(R/\bm{\omega})\).

In our treatment, quadrature expectations and coherent state
amplitudes are unitless and scaled so that the vacuum expectation of
the square of a normalized quadrature is \(1/2\). The units of the
weights \(\bm{\omega}\) and of the scale \(R\) are therefore
identical. In the case of calorimeters, they both have the same energy
units. For the purposes of BBP homodyne, quantities of interest
depend on \(R\) and \(\bm{\omega_{k}}\) only through ratios such
  as \(R/\omega_{k}\).

For analyzing the behavior of BBP homodyne, we define $\delta=1/R$ and introduce the BBP measurement operator
\begin{align} \label{observable}
\hat{q}_{\bm{\alpha},\delta}
&= \delta \Delta \hat{E} \nonumber\\
&= \hat{q}_{\bm{a}, \bm{\alpha}}
+ \delta
\qty(
\qty(\bm{\omega}*\hat{\bm{a}}^{\dagger}\vphantom{\hat{\bm{b}}^{\dagger}})\cdot\hat{\bm{b}}
+ \qty(\bm{\omega}*\hat{\bm{b}}^{\dagger})\cdot\hat{\bm{a}}
).
\end{align}
As noted previously \(\Delta\hat{E}\) is defined as an operator on the
incoming modes and therefore depends on \(\delta\) through the
displacement on the LO modes.  The associated LO amplitudes are
\begin{align}
\bm{\beta}_{\LO,\bm{\alpha},\delta} &=
-i\bm{\alpha}*\frac{1}{\bm{\omega}\delta}.\label{eq:lodisp}
\end{align}
We define $\hat{q}_{\bm{\alpha},0}= \hat{q}_{\bm{a},\bm{\alpha}}$.
When \(\bm{\alpha}\) is clear from context, we abbreviate
\(\hat{q}_{\delta}=\hat{q}_{\bm{\alpha},\delta}\) and
\(\hat{q}=\hat{q}_{\bm{\alpha},0}\).

The BPP measurement operator $\hat{q}_{\delta}$ is by design
proportional to a difference of two commuting total energy operators
at the outgoing modes, so the spectrum of $\hat{q}_{\delta}$ is
discrete. The outgoing mode Fock states that diagonalize the energy
operators correspond to displaced Fock states in the incoming modes,
where the displacement diverges to infinity as \(\delta\) goes to
zero. In contrast, \(\hat{q}\) has continuous spectrum and improper
eigenspaces associated with its spectral measure. This immediately
suggests that the convergence of \(q_{\delta}\) to the target
quadrature \(q_{0}\) is not straightforward. The BBP measurement
operator is unbounded, as is the quadrature measurement. As a result,
there are states in the Hilbert space not in the domains of these
operators, and for such states convergence is impossible. Thus proofs
and quantification of convergence are necessarily state-dependent.

The BBP homodyne configuration, as described above, has perfect
calorimeters. However, as with the standard homodyne configuration,
noisy but high-efficiency calorimeters can be used, provided that the
noise scales sublinearly with total energy. We do not include such
noise in our analysis.

\section{Weak convergence to a quadrature measurement}
\label{sec:wconv}
Our first convergence result establishes that the moments for BBP
homodyne measurement outcome distributions converge to the moments of
the target quadrature. This result is well known for standard
homodyne, for example, see Ref.~\cite{Sanders}, Eq. (3.2) and
following, or Ref.~\cite{kiukas2008moment_jmo} Prop. 8. The proofs for
standard homodyne generalize with little modification to BBP homodyne.
We use the angle-bracket notation $\lang \cdot \rang_{\rho}$ for
expectations of operators with respect to state $\rho$. We omit the
subscript when the state with respect to which expectations are
computed is clear from context.

\begin{theorem}\label{thm:momentlim}
  Suppose the state $\rho_{\bm{a}}$ on the signal modes $\bm{a}$ has
  well-defined expectations for all polynomials in the mode operators
  and their adjoints up to degree $n$, and the joint initial state is
  $\rho_{\bm{a}}\otimes\bm{0}_{\bm{b}}$.  Then the moments of
  $\hat{q}_{\delta}$ and $\hat{q}$ satisfy
\begin{align}
\lang \hat{q}_{\delta}^{n}\rang  &= \lang
\hat{q}^{n} \rang + O(\delta^{2}), \label{thm1:eq}
\end{align} where the constant in the order notation depends on
  $\rho_{\bm{a}}$.
\end{theorem}

\begin{proof}
  Define the operator
  \(\hat{C}=(\hat{q}_{\delta}-\hat{q})/\delta=\qty(\bm{\omega}*\hat{\bm{a}}^{\dagger}\vphantom{\hat{\bm{b}}^{\dagger}})\cdot\hat{\bm{b}}
  + \qty(\bm{\omega}*\hat{\bm{b}}^{\dagger})\cdot\hat{\bm{a}}\), which
  does not depend on $\delta$. We expand
  $\hat{q}_{\delta}^{n}= ( \hat{q} + \delta\; \hat{C} )^{n}$ as a sum
  of $2^{n}$ monomials expressed as ordered products of $\hat{q}$ and
  $\delta \; \hat{C}$. We can order the monomials by the power of
  $\delta$ that multiplies them. Let $\hat{Q}_{k}$ be the sum of the
  monomials that are multiplied by $\delta^{k}$, so
  $\hat{q}_{\delta}^n=\sum_{k=0}^n\hat{Q}_k$. Then
  $\hat{Q}_{0}= \hat{q}^{n}$, so
  $\lang \hat{Q}_{0}\rang = \lang \hat{q}^{n}\rang$ contributes the
  first term on the right-hand side of Eq.~\eqref{thm1:eq}. We can
  express $\hat{Q}_{1}$ as
  \begin{align}
    \hat{Q}_{1} &= \delta \sum_{k=0}^{n-1} \hat{q}^{k} \hat C \hat{q}^{n-k-1}.
  \end{align}
  The factors of $\hat{q}$ in the sum act only on the signal modes
  $\bm{a}$, and $\hat{C}$ is linear in the LO mode operators. Since
  $\rho$ is vacuum on the pre-displacement LO modes, the expectations
  of the summands of $\hat{Q}_{1}$ are zero. Consequently
  $\lang \hat{Q}_{1}\rang=0$. The remaining terms in the expansion are
  multiplied by a factor of order $\delta^{2}$ or smaller. Their
  expectations may be expressed as expectations of monomials of degree
  at most $n$ in the signal mode operators and their adjoints, and of
  degree at most $n$ in the LO mode operators and their
  adjoints. Taking account of the vacuum state in the LO mode, these
  expectations reduce to expectations of $\rho_{\bm{a}}$ of monomials
  of degree at most $n$ in the signal mode operators. The expectations
  of these monomials are finite by assumption and do not depend on
  $\delta$. Thus
  \begin{align}
    \lang \hat{q}_{\delta}^{n}\rang
    &=
      \sum_{k=0}^{n}\lang \hat{Q}_{k}\rang \nonumber\\
    &= \lang \hat{Q}_{0}\rang + \lang \hat{Q}_{1}\rang+ O(\delta^{2})\nonumber\\
    &= \lang\hat{q}^{n}\rang + O(\delta^{2}).
  \end{align}
\end{proof}

We explicitly evaluate the differences between the  moments of $\hat{q}_{\delta}$ and $\hat{q}$ for the first and second moment. The first moment does not depend on $\delta$ and is identical to the target quadrature's. Because the pre-displacement LO modes are in vacuum,
\begin{align}
\langle \hat{q}_{\delta} - \hat{q} \rangle
&= \delta
\sum_{k=1}^{N} \omega_{k} \langle \hat{a}^{\dagger}_{k} \hat{b}_{k} + \hat{a}_{k} \hat{b}^{\dagger}_{k} \rangle\nonumber\\
&= 0.
\end{align}
For the second moment we have
\begin{align}
\langle \hat{q}_{\delta}^{2} - \hat{q}^{2} \rangle
& =
\delta^{2} \sum_{k,k'=1}^{N} \omega_{k} \omega_{k'} \left\langle (
\hat{a}^{\dagger}_{k} \hat{b}_{k} + \hat{a}_{k} \hat{b}^{\dagger}_{k}
) ( \hat{a}^{\dagger}_{k'} \hat{b}_{k'} + \hat{a}_{k'}
\hat{b}^{\dagger}_{k'} ) \right\rangle \nonumber\\
&\phantom{=\ }
- \delta \sum_{k=1}^{N} \omega_{k} \left\lang
(\hat{a}^{\dagger}_{k} \hat{b}_{k} + \hat{a}_{k} \hat{b}^{\dagger}_{k}
) \hat{q}\right\rang - \delta \sum_{k=1}^{N} \omega_{k} \left\lang \hat{q} ( \hat{a}^{\dagger}_{k}
\hat{b}_{k} + \hat{a}_{k} \hat{b}^{\dagger}_{k} )\right\rang\nonumber\\
&=
\delta^{2} \sum_{k} \omega^{2}_{k} \langle \hat{n}_{a_k} \rangle,
\end{align}
where all terms other than those with factors of the form
$\hat{b}_{k}\hat{b}_{k}^{\dagger}$ are zero because the
pre-displacement LO modes are in vacuum. If it is known that the
signal state is Gaussian, then it is sufficient to measure the first
and second moments of its quadratures.  The above shows that the
second moments of $\hat{q}_{\delta}$ are biased high by a term of
order $\delta^{2}$ with a coefficient that can be estimated knowing
only bounds on the expected photon numbers.

Theorem~\ref{thm:momentlim} applies specifically to states
with well-defined expectations for all polynomials of mode operators.  A
dense linear space of pure states with such well-defined expectations
is the set \(\cD_{S}\) of Schwartz states, defined as states whose
Wigner functions decay
superpolynomially~\cite{hernandez2022rapidly}. This set of states is
preserved by all polynomials of the mode operators and their adjoints.
It includes number states and coherent states and their finite linear
combinations.
 
Next we show that the moment convergence for a restricted family of
signal states in Thm.~\ref{thm:momentlim} implies that for every
bounded continuous function \(f\) of the reals and every signal state,
the expectations of \(f(\hat{q}_\delta)\) with respect to the
BBP homodyne outcome distributions converge to the expectation of
\(f(\hat{q})\).  This property is equivalent to weak convergence in
the sense of probabilities of the positive operator valued measures
(POVMs) realized by BBP homodyne to the spectral measure of
\(\hat{q}\) (see~\cite{kiukas2008moment_jmo} Prop. 3).
  
\begin{theorem}\label{thm:weaklim}
  For every continuous complex-valued bounded function $f$ on the
  reals and for every  state
  $\rho=\rho_{\bm{a}}\otimes \bm{0}_{\bm{b}}$, we have
  \begin{align}
    \lim_{\delta\rightarrow 0}\lang f(\hat{q}_{\delta})\rang
    &= \lang f(\hat{q}) \rang.
      \label{eq:weaklim}
  \end{align}
\end{theorem}

\begin{proof}
  To prove the theorem we implement a version of the sequence of steps
  given in Sect.~4 of Ref.~\cite{kiukas2008moment_jmo} for
  establishing weak convergence for measurement schemes. In this
  reference, the steps are implemented to prove weak convergence of
  standard homodyne with one mode.  The first step is to identify a
  dense subspace of states on which the target quadrature's
  measurement outcome distribution is determined by its moments
  (Ref.~\cite{kiukas2008moment_jmo} Def. 2). A probability
  distribution $\mu$ on the reals is determined by its moments if for
  every probability distribution $\nu$ whose moments are the same as
  those of $\mu$, we have $\mu=\nu$. The set
  $\cD_{\text{coh}}\subseteq \cD_{S}$ of finite linear combinations of
  coherent states suffices for this purpose. For one mode, this is a
  consequence of Ref.~\cite{kiukas2008moment_jmo} Lemma 2.  Because
  $\hat{q}$ is the quadrature of one mode, this Lem. 2 suffices for
  BBP homodyne. See also the discussion after Proposition 2 in the
  reference. The next step is to verify that for states in
  \(\cD_{\text{coh}}\), the positive-operator-valued measures (POVMs)
  associated with $\hat{q}_{\delta}$ have moments converging to those
  of $\hat{q}$. For BBP homodyne, since \(\cD_{\text{coh}}\) is a
  subset of the set of Schwartz states, this is a consequence of
  Thm.~\ref{thm:momentlim}. For the purpose of applying the results of
  Ref.~\cite{kiukas2008moment_jmo}, the POVMs associated with
  $\hat{q}_{\delta}$ are the POVMs on the signal modes obtained from
  the projection-valued measures of the operators \(\hat{q}_{\delta}\)
  by fixing the pre-displacement state of the LO modes to be vacuum.
  The POVM for $\hat{q}$ is projection-valued, but the POVMs for
  $\hat{q}_{\delta}$ are not, they are positive-operator valued and
  referred to as ``semispectral measures'' in
  Ref.~\cite{kiukas2008moment_jmo}. With this, the conditions of
  Ref.~\cite{kiukas2008moment_jmo} Prop. 5 are satisfied. That is,
  with the definitions of this reference, because of moment
  convergence, there is a POVM that is a moment limit of the POVMs
  associated with \(\hat{q}_{\delta}\). One such moment limit is the
  spectral measure of \(\hat{q}\).  Since the latter is determined by
  its moments, this moment limit is unique. The conclusion from the
  reference's Prop. 5 is that the POVMs associated with
  \(\hat{q}_{\delta}\) converge weakly in the sense of probabilities
  to the spectral measure of \(\hat{q}\). This is equivalent to the
  conclusion of our theorem.
\end{proof}

There are many equivalent definitions for weak convergence of the
measures associated with \(\hat{q}_{\delta}\) to the measures of
\(\hat{q}\).  The version given in Thm.~\ref{thm:weaklim} corresponds
to Ref.~\cite{kiukas2008moment_jmo} Prop. 3 (iv), which is also
equivalent to the convergence of overlaps as expressed in the
following corollary. 

\begin{corollary}
  Let \(\bm{g}\) be any family of modes or other
  quantum systems that are not involved in the BBP homodyne
  measurements, and \(f\) a continuous complex-valued bounded function
  on the reals.  For all joint pure states $\ket{\phi}$ and
  $\ket{\psi}$ of \(\bm{g}\), the signal, and the LO modes, if
  \(\ket{\phi}\) and \(\ket{\psi}\) are vacuum on the LO modes, then
  we have
  \begin{align}
    \lim_{\delta\rightarrow 0}\bra{\phi}f(\hat{q}_{\delta})\ket{\psi}
    &= \bra{\phi}f(\hat{q})\ket{\psi}.
  \end{align}
\end{corollary}

\begin{proof}
  For any complex, bilinear form $\bra{\phi}\hat{A}\ket{\psi}$ with
  bounded operator $\hat{A}$,
  \begin{align}
    \bra{\phi}\hat{A}\ket{\psi}
    &=
      \frac{1}{4}
      \sum_{j=0,1,2,3} (-i)^{j}(\bra{\phi}+(-i)^{j}\bra{\psi})\hat{A}(\ket{\phi}+i^{j}\ket{\psi})
      \nonumber\\
    &= \frac{1}{4}\sum_{j=0,1,2,3} (-i)^{j}\bra{\rho_{j}}\hat{A}\ket{\rho_{j}},
      \label{eq:polarization}
  \end{align}  
  where \(\ket{\rho_{j}}=\ket{\phi}+i^{j}\ket{\psi}\).   We can
    re-express
    \(\bra{\rho_{j}}\hat{A}\ket{\rho_{j}} =
    \tr(\ketbra{\rho_{j}}\hat{A})\).  After substituting
    \(\hat{A}=f(\hat{q}_{\delta})\) and tracing out systems other than
    the signal and LO modes in the state \(\ketbra{\rho_{j}}\), we can
    apply Thm.~\ref{thm:weaklim} to complete the proof of the
    corollary.  The identity in Eq.~\eqref{eq:polarization} is an instance of
    the polarization identity, a textbook method used to reconstruct
    an inner product from the associated norm, for instance see
    \cite{kadison:qf1997a,ReedSimon}.
\end{proof}

\section{Discussion}
\label{sec:discussion}

BBP homodyne can be seen to be a generalization of standard pulsed
homodyne by setting all the weights to \(\omega_k=1\). The well-known
properties of standard pulsed homodyne are preserved. In particular,
the moments of the BBP homodyne observables converge to the moments of
the target quadrature and the distributions of BBP measurement
outcomes converge weakly to that of the target quadrature.

The main motivation for introducing BBP homodyne is as a way of taking
advantage of calorimeters to measure quadratures of broadband modes
such as those present in optical femtosecond pulses or modes of
interest in quantum field theory such as Rindler modes. The BBP
homodyne configuration preserves the simplicity of the standard
homodyne experimental configuration with only two detectors after one
beamsplitter. The detectors need not resolve time and can be slow. An
alternative approach to homodyne measurements of broadband modes is to
use a device such as a Bragg grating to split the incoming modes
according to their wavelengths, then combine standard, narrowband
homodyne measurements at each wavelength. A version of this technique
is proposed in Ref.~\cite{shaked2018lifting}. This technique has the
advantage of being able to simultaneously measure multiple quadratures
across the resolved spectrum, at the cost of a more complicated
experimental configuration.

We have not yet considered the effect of energy-dependent inefficiency
on the performance of BBP homodyne. For standard pulsed homodyne,
detector inefficiency can be taken into account with an operational
theory of homodyne~\cite{Banaszek}. It is also of interest to directly
determine and exploit the effective POVMs at finite LO amplitudes as
done, for example in Ref.~\cite{Sanders}. Going further, when the
calorimeters perform well for low incident energy, at small LO
amplitude and signal energies, one can directly take advantage of each
calorimeter's measurement outcome, generalizing or bypassing the
subtraction, as done for weak-field homodyne in
Refs.~\cite{vogel1995homodyne,avagyan2023quantum}.

For applications of BBP homodyne, it is desirable to quantify the
convergence of the measurement outcome distribution to the ideal one
for the target quadrature. In many applications, the measurement
outcome is used in conditional operations such as displacements of
unmeasured modes.  Examples include CV quantum teleportation~\cite{braunstein1998teleportation,furusawa1998teleportation}
and CV quantum computing~\cite{gottesman2001encoding}. For these applications, it will be
helpful to quantify the fidelity of the conditional operations.

For the application of BBP homodyne to characterizing the state of
Rindler modes and verifying the thermal states of these modes, it is
necessary to determine how to realize the local oscillator and
beamsplitter while maintaining the compatibility with the relativistic
quantum field under investigation. Because it is impossible to measure
standard, single-frequency Rindler modes, it is also necessary to
determine how to reveal the desired quadrature information from
smeared such modes. Relevant suggestions have been offered for
different detection systems in~\cite{PRL2011}.

\section{Acknowledgment}
We thank Adriana Lita for help with identifying and describing the
performance of broadband TES detectors.  We also thank Zachary Levine
and Michael Mazurek for assistance with reviewing the paper before
submission.  E. S and J. R. v. M. acknowledge support from the
Professional Research Experience Program (PREP) operated jointly by
NIST and the University of Colorado. This work includes contributions
of the National Institute of Standards and Technology, which are not
subject to U.S. copyright.

\bibliography{Contents}

\end{document}